\documentclass[a4paper,fleqn]{cas-sc}

\usepackage[numbers,sort&compress]{natbib}
\usepackage[T1]{fontenc}
\usepackage{amsmath,amssymb,mathtools,bbm,amsthm}
\usepackage{tikz}
\usepackage{float}
\usetikzlibrary{arrows.meta,positioning}
\numberwithin{equation}{section}

\theoremstyle{plain}
\newtheorem{thm}{Theorem}[section]
\newtheorem{prop}[thm]{Proposition}
\newtheorem{cor}[thm]{Corollary}

\theoremstyle{definition}
\newtheorem{defi}[thm]{Definition}
\theoremstyle{remark}
\newtheorem{rem}[thm]{Remark}

\newcommand{\Out}{\mathrm{Out}}

\newcommand{\lab}{\mathrm{lab}}

\newcommand{\tr}{\mathrm{tr}}
\newcommand{\Fib}{\mathrm{Fib}}
\newcommand{\run}{\mathrm{run}}
\newcommand{\Phase}{\mathrm{Phase}}
\newcommand{\holdsucc}{\mathrm{hsucc}}

\begin{document}
\let\WriteBookmarks\relax
\def\floatpagepagefraction{1}
\def\textpagefraction{.001}

\shorttitle{Destination-Labeled Self-Looping Systems with Dwell}
\shortauthors{Reda Belaiche}

\title[mode=title]{Destination-Labeled Self-Looping Systems with Dwell: Intrinsic Characterization, Realization Cost, and Recognition}

\author[1]{Reda Belaiche}[orcid=0000-0002-8741-0349]
\cormark[1]
\ead{reda.belaiche@u-pec.fr}
\credit{Conceptualization, Methodology, Formal analysis, Writing - original draft, Writing - review \& editing}

\affiliation[1]{organization={Department of Computer Science, University Institute of Technology of Cr\'eteil-Vitry, Paris-Est Cr\'eteil University},
                addressline={122 rue Paul Armangot},
                city={Vitry-sur-Seine},
                postcode={94400},
                country={France}}

\cortext[1]{Corresponding author}

\begin{abstract}
Many physical state-transition systems --- machinery wear cycles, human activity sequences, or physiological progressions --- are naturally modeled by per-state classifiers rather than by a single global sequence model. Such architectures require a control skeleton that enforces hard graph constraints and minimum residence times. We study that skeleton in the form of destination-labeled self-looping systems with dwell (DLSL systems).

Once dwell is imposed, visible states no longer suffice: two histories may end in the same visible state while differing in whether departure is already enabled. The structural question is therefore intrinsic: which deterministic transducers arise from DLSL phase expansion over a fixed visible graph? We show first that the phase-expanded realizations of DLSL systems are exactly the fiber-linear graph-respecting transducers. Second, under reachability and realizable-departure hypotheses, equivalent accessible fiber-linear transducers over the same visible graph are isomorphic, so the visible transduction determines the dwell vector and local decision maps uniquely. Third, enforcing dwell values $(d_i)$ requires exactly $\sum_i d_i$ control states in the deterministic graph-preserving setting. Recognition and reconstruction are polynomial-time, in $O(|Q||\Omega|)$ time. We also treat an edge-entry extension in which decisions may enter designated interior phases of successor fibers; the same path-fiber analysis yields the corresponding converse and recognition results.
\end{abstract}


\maketitle

\section{Introduction and main structural results}

Many physical processes --- machinery wear cycles, human activity sequences, and physiological state progressions among them --- evolve through a discrete set of qualitatively distinct regimes. A natural modeling choice in such settings is to assign a separate classifier to each regime, trained on data generated within that regime, rather than to train a single global model on the full sequence. This per-state classifier structure requires a control skeleton specifying which regime changes are physically admissible and how long the system must remain in a regime before a departure is credible. The DLSL model is exactly that skeleton.

A global sequence model --- whether based on hidden-state inference, recurrent neural networks, or semi-Markov discrimination --- learns transition structure from data, either explicitly or implicitly. When the transition graph and minimum dwell constraints are physically known in advance, it is more natural to encode them as hard structural constraints rather than as regularities that the learning algorithm must rediscover. That yields a model that is easier to interpret and, in principle, more data-efficient. The present framework isolates the finite-state control structure underlying that design: hard graph constraints, minimum dwell, per-state decisions, and deterministic execution in a single symbolic object.

Once minimum dwell is imposed, visible states no longer suffice to describe the controller exactly. Two histories may end in the same visible state while differing in whether departure is already enabled. The standard resolution is to refine each visible state into a short internal chain that records the remaining forced-hold depth. This yields the familiar phase-expanded realization. The forward construction is straightforward. The main question of the paper is the converse one: \emph{which deterministic transducers arise in exactly this way over a fixed visible graph?}

The paper is organized around three questions. The first is constructive: how does one compile a DLSL system into a deterministic transducer, and what is the exact graph-preserving state cost? Sections~\ref{sec:realization_preliminaries} and~\ref{sec:exact_state_cost} answer this; the compilation is immediate, and the lower-bound argument shows that the cost is optimal. The second question is intrinsic: which deterministic transducers arise from this construction without phase coordinates being assumed in advance? Section~\ref{sec:intrinsic_characterization} answers this, establishing fiber-linearity as the complete characterization and rigidity as its structural consequence. The third question is algorithmic: can fiber-linearity be tested efficiently? Section~\ref{sec:recognition} answers yes, in $O(|Q||\Omega|)$ time, and reconstructs the underlying DLSL data. Section~\ref{sec:edge_entry_extension} extends the characterization to the edge-entry variant, in which decisions may enter designated interior phases of successor fibers. The visible-state trace formulation is collected in Appendix~\ref{app:runlength}.

\section{Symbolic model, running example, and normal form}\label{sec:model}

We now formalize the symbolic skeleton and fix the notation used throughout. The development is purely symbolic, but the intended motivation remains the same: each local map \(g_i\) may be viewed as the output of a per-state classifier, while the DLSL structure supplies the hard control constraints. No statistical assumptions are needed in the formal theory.

\subsection{Destination-labeled self-looping visible graphs}

Let \(S\) be a finite visible-state set and let \(E \subseteq S \times S\) be a directed graph of admissible visible moves. Let \(\Sigma\) be a finite alphabet and let \(\ell:S\to\Sigma\) be a visible-state labeling.

\begin{defi}[Destination-labeled graph]
A directed graph \((S,E)\) equipped with \(\ell:S\to\Sigma\) is \emph{destination-labeled} if every admissible edge \((i,j)\in E\) carries the label \(\ell(j)\). Equivalently, if we define
\[
\lab(i\to j):=\ell(j),
\]
then the label of an edge depends only on its destination.
\end{defi}

\begin{defi}[Self-looping]
A destination-labeled graph is \emph{self-looping} if every visible state has a self-loop:
\[
(i,i) \in E
\qquad \text{for all } i \in S.
\]
\end{defi}

\begin{defi}[Label-deterministic out-neighborhood]
A destination-labeled graph is \emph{label-deterministic} if for every visible state \(i\) and every label \(a\in\Sigma\),
\[
\left|\{j \in \Out(i): \ell(j)=a\}\right| \le 1.
\]
\end{defi}

Under label-determinism, a visible successor is uniquely identified by its destination label. We therefore define
\[
\Sigma_i := \{\ell(j): j \in \Out(i)\},
\]
together with the induced partial transition map
\[
\delta(i,a)=
\begin{cases}
j, & \text{if } j \in \Out(i) \text{ and } \ell(j)=a,\\
\text{undefined}, & \text{otherwise.}
\end{cases}
\]

\subsection{Running example: activity recognition with minimum dwell}

We use a small running example throughout to keep the notation concrete. Let
\[
S=\{\mathrm{rest},\mathrm{walk},\mathrm{run}\}.
\]
The visible graph allows
\[
\begin{aligned}
&\mathrm{rest}\to\mathrm{rest},\quad
\mathrm{rest}\to\mathrm{walk},\quad
\mathrm{walk}\to\mathrm{rest},\quad
\mathrm{walk}\to\mathrm{walk},\\
&\mathrm{walk}\to\mathrm{run},\quad
\mathrm{run}\to\mathrm{walk},\quad
\mathrm{run}\to\mathrm{run},
\end{aligned}
\]
and forbids direct \(\mathrm{rest}\to\mathrm{run}\) and \(\mathrm{run}\to\mathrm{rest}\) transitions. The dwell values are
\[
d_{\mathrm{rest}}=3,\qquad d_{\mathrm{walk}}=2,\qquad d_{\mathrm{run}}=2.
\]

One may view the visible states as symbolic activity regimes and the local maps \(g_i\) as regime-dependent decision rules. This interpretation motivates the model; no statistical assumptions enter the formal development.

The example is deliberately small. Its phase-expanded realization has \(3+2+2=7\) control states, so it is convenient for illustrating the constructions and converse results below.

It is useful to distinguish the visible DLSL skeleton from its exact deterministic realization. The former has one visible state per label and records only the destination-labeled transition structure; the latter refines each visible state into a phase chain in order to make dwell memory explicit.

\begin{figure}[pos=H]
\centering
\begin{tikzpicture}[
    >=Latex,
    node distance=28mm and 28mm,
    state/.style={draw, circle, minimum size=11mm, inner sep=0pt},
    lbl/.style={font=\small}
]

\node[state] (rest) {\textsf{rest}};
\node[state, right=of rest] (walk) {\textsf{walk}};
\node[state, right=of walk] (run) {\textsf{run}};

% self-loops
\draw[->, loop above] (rest) to node[above,lbl] {rest} ();
\draw[->, loop above] (walk) to node[above,lbl] {walk} ();
\draw[->, loop above] (run) to node[above,lbl] {run} ();

% visible transitions
\draw[->] (rest) -- node[above,lbl] {walk} (walk);
\draw[->, bend left=20] (walk) to node[below,lbl] {rest} (rest);

\draw[->] (walk) -- node[above,lbl] {run} (run);
\draw[->, bend left=20] (run) to node[below,lbl] {walk} (walk);

\end{tikzpicture}
\caption{Implicit-dwell DLSL skeleton for the running example. There is one visible state per label, and every transition entering a state carries that state's label. Minimum dwell is part of the symbolic semantics, not of the visible graph itself.}
\label{fig:running_example_implicit_dwell}
\end{figure}
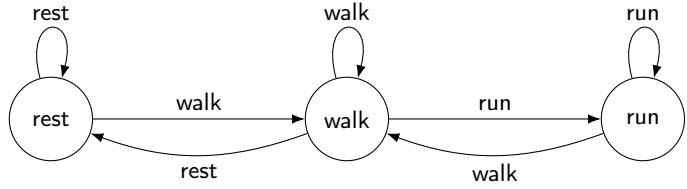

\autoref{fig:running_example_implicit_dwell} shows the visible DLSL skeleton of the running example, before any phase expansion is introduced. There is exactly one visible state for each label, and the labeling is destination-based: every transition entering \(\textsf{rest}\), \(\textsf{walk}\), or \(\textsf{run}\) carries the corresponding destination label. The minimum dwell values \(d_{\mathrm{rest}}=3\), \(d_{\mathrm{walk}}=2\), and \(d_{\mathrm{run}}=2\) are not encoded by extra visible states in this figure. Instead, they are enforced by the operational semantics through the residual dwell phase.

This distinction is important. The visible DLSL skeleton records the admissible labeled moves between regimes, while the phase-expanded realization introduced later makes the dwell memory explicit in the control state. In particular, when some \(d_i>1\), the one-state-per-label skeleton is not by itself an exact deterministic realization of the dwell-constrained behavior.

\subsection{Symbolic DLSL systems with dwell}

Fix a finite symbolic input alphabet \(\Omega\), an initial visible state \(s_0 \in S\), and a dwell map
\[
d:S \to \mathbb N_{\ge 1}, \qquad i \mapsto d_i.
\]
For each visible state \(i\), let
\[
g_i : \Omega \to \Sigma_i
\]
be a local symbolic decision map.

\begin{defi}[Symbolic DLSL system with dwell]
A \emph{symbolic destination-labeled self-looping system with dwell} consists of the data
\[
\mathcal A = (S,E,\ell,s_0,d,(g_i)_{i\in S})
\]
such that \((S,E,\ell)\) is destination-labeled, self-looping, and label-deterministic. Its \emph{operational state} is a pair \((i,r)\), where \(i\in S\) is the visible state and
\[
r \in \{0,1,\dots,d_i-1\}
\]
is the residual dwell phase.
\end{defi}

The operational update on input \(u\in\Omega\) is
\begin{equation}
\label{eq:operational_update}
(i,r) \xrightarrow{u}
\begin{cases}
(i,r-1), & r>0,\\
(\delta(i,g_i(u)),\, d_{\delta(i,g_i(u))}-1), & r=0.
\end{cases}
\end{equation}
Thus, while \(r>0\), the machine is forced to remain in the current visible state. Once \(r=0\), it may apply the local decision map and move to the unique visible successor identified by the predicted destination label.

In the running example, the operational state \((\mathrm{walk},1)\) means: the current visible regime is \(\mathrm{walk}\), and one forced-hold step remains before any proposal to leave \(\mathrm{walk}\) can be acted upon. By contrast, \((\mathrm{walk},0)\) means that the next input may either keep the machine in \(\mathrm{walk}\) or send it to \(\mathrm{rest}\) or \(\mathrm{run}\), provided the corresponding destination label is proposed by \(g_{\mathrm{walk}}\).

For an input word \(w=u_1\cdots u_n \in \Omega^*\), let
\[
(y_0,r_0)=(s_0,d_{s_0}-1)
\]
and define recursively
\[
(y_t,r_t) = \delta_{\mathrm{op}}((y_{t-1},r_{t-1}),u_t), \qquad t=1,\dots,n.
\]
We write
\[
\tr_{\mathcal A}(w):=y_0 y_1 \cdots y_n \in S^{n+1}
\]
for the visible-state trace and
\[
\Out_{\mathcal A}(w):=y_1 \cdots y_n \in S^n
\]
for the post-initial output word.

\begin{defi}[Standing assumptions]
\label{def:standing_assumptions}
The symbolic DLSL system is:
\begin{enumerate}
    \item \emph{destination-complete} if, for every visible state \(i\) and every label \(a\in\Sigma_i\), there exists an input symbol \(u\in\Omega\) such that \(g_i(u)=a\);
    \item \emph{reachable} if every visible state under discussion is reachable from \(s_0\) by some input word.
\end{enumerate}
\end{defi}

Destination-completeness is only needed when we want exact visible-trace realizability statements and lower bounds. It is not needed for the basic compilation into a deterministic transducer.

Because the visible graph is self-looping and destination-labeled, every visible state automatically has its own label available as a hold action. Concretely, if \((i,i)\in E\), then the self-loop into \(i\) carries label \(\ell(i)\), so \(\ell(i)\in\Sigma_i\). We use this simple observation only implicitly later, when a decision phase needs to realize a visible self-loop without augmenting the visible graph.

\section{Realization preliminaries}\label{sec:realization_preliminaries}

\subsection{Phase expansion into a deterministic transducer}

The forward realization is standard: once the residual dwell phase is carried explicitly, one obtains a deterministic finite-state controller immediately. We record it only to fix notation and to make the state count explicit before turning to lower bounds and intrinsic characterization.

\begin{prop}[Phase expansion into a deterministic finite-state transducer]
\label{thm:compilation_transducer}
For every symbolic DLSL system with dwell, there exists a deterministic Mealy transducer
\[
T_{\mathcal A}=(Q,q_0,\Omega,S,\Delta,\omega)
\]
such that for every input word \(w \in \Omega^*\),
\[
\omega^*(q_0,w)=\Out_{\mathcal A}(w).
\]
Moreover, one may choose
\[
Q = \{(i,r): i \in S,\ 0 \le r \le d_i-1\},
\qquad
q_0=(s_0,d_{s_0}-1),
\]
so that
\[
|Q| = \sum_{i\in S} d_i.
\]
\end{prop}

\begin{proof}
The construction is the phase expansion itself. Let
\[
Q = \{(i,r): i \in S,\ 0 \le r \le d_i-1\},
\qquad
q_0=(s_0,d_{s_0}-1),
\]
and define
\[
\Delta((i,r),u)=
\begin{cases}
(i,r-1), & r>0,\\
(\delta(i,g_i(u)),\, d_{\delta(i,g_i(u))}-1), & r=0,
\end{cases}
\]
together with
\[
\omega((i,r),u)=
\begin{cases}
i, & r>0,\\
\delta(i,g_i(u)), & r=0.
\end{cases}
\]
These formulas are exactly the operational semantics \eqref{eq:operational_update}, with the next visible state emitted after each input symbol. Hence, by a direct induction on the input length, the transducer state after any prefix coincides with the operational state of the symbolic system, and the emitted word is precisely \(\Out_{\mathcal A}(w)\). The determinism of the transducer follows from the fact that each \(g_i\) is a function and label-determinism makes \(\delta(i,g_i(u))\) unique. Finally,
\[
|Q|=\sum_{i\in S} d_i
\]
by construction. \qedhere
\end{proof}

\begin{figure}[pos=H]
\centering
\begin{tikzpicture}[
    >=Latex,
    node distance=13mm and 22mm,
    state/.style={draw, circle, minimum size=11mm, inner sep=0pt},
    lbl/.style={font=\small}
]

% --- rest fiber ---
\node[state] (r2) {$\textsf{rest}$};
\node[state, below=of r2] (r1) {$\textsf{rest}$};
\node[state, below=of r1] (r0) {$\textsf{rest}$};

% --- walk fiber ---
\node[state, right=34mm of r1] (w1) {$\textsf{walk}$};
\node[state, below=of w1] (w0) {$\textsf{walk}$};

% --- run fiber ---
\node[state, right=34mm of w1] (u1) {$\textsf{run}$};
\node[state, below=of u1] (u0) {$\textsf{run}$};

% phase annotations
\node[left=1mm of r2, lbl] {\scriptsize 2};
\node[left=1mm of r1, lbl] {\scriptsize 1};
\node[left=1mm of r0, lbl] {\scriptsize 0};

\node[left=1mm of w1, lbl] {\scriptsize 1};
\node[left=1mm of w0, lbl] {\scriptsize 0};

\node[left=1mm of u1, lbl] {\scriptsize 1};
\node[left=1mm of u0, lbl] {\scriptsize 0};

% intra-fiber dwell transitions
\draw[->] (r2) -- node[left,lbl] {rest} (r1);
\draw[->] (r1) -- node[left,lbl] {rest} (r0);

\draw[->] (w1) -- node[right,lbl] {walk} (w0);
\draw[->] (u1) -- node[right,lbl] {run} (u0);

% self-loops on decision states
\draw[->, loop below] (r0) to node[below,lbl] {rest} ();
\draw[->, loop below] (w0) to node[below,lbl] {walk} ();
\draw[->, loop below] (u0) to node[below,lbl] {run} ();

% inter-fiber transitions
\draw[->] (r0) -- node[above,lbl,xshift=-14pt,yshift=-22pt] {walk} (w1);
\draw[->, bend left=20] (w0) to node[below,lbl] {rest} (r2);

\draw[->] (w0) -- node[below,lbl,xshift=-20pt,yshift=-18pt] {run} (u1);
\draw[->] (u0) -- node[below,lbl,xshift=32pt,yshift=-2pt] {walk} (w1);

\end{tikzpicture}
\caption{A destination-labeled realization of the running example with visible dwell phases. Each fiber is associated with one visible label: every transition entering a \(\textsf{rest}\)-state is labeled \(\textsf{rest}\), every transition entering a \(\textsf{walk}\)-state is labeled \(\textsf{walk}\), and every transition entering a \(\textsf{run}\)-state is labeled \(\textsf{run}\). The phase-\(0\) states are the decision states; the higher-phase states are forced-hold states. This is the reset-to-top canonical form; the edge-entry variant, in which decisions may enter a designated interior state of the successor fiber rather than its top state, is described in Section~\ref{sec:edge_entry_extension} and illustrated in Figure~\ref{fig:edge_entry_reset}.}
\label{fig:running_example_dlsl}
\end{figure}
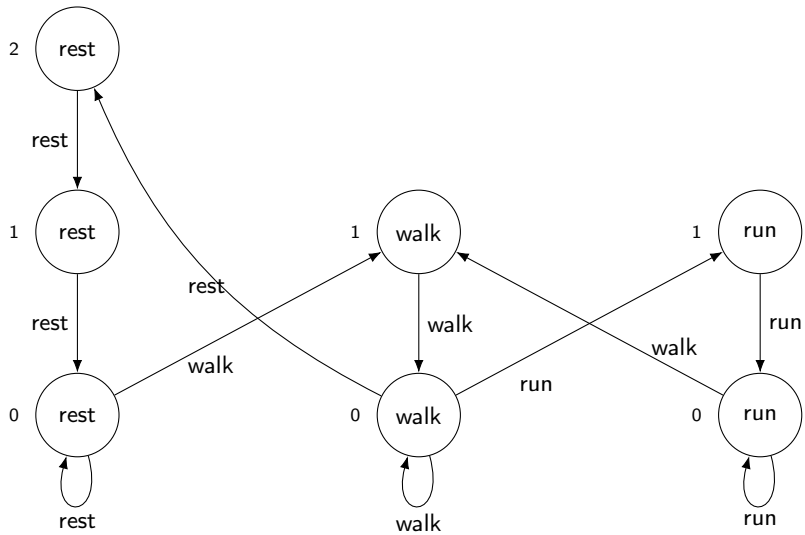

In \autoref{fig:running_example_dlsl}, the destination-labeled structure and the dwell refinement are visible at the same time. The automaton is organized into three fibers, corresponding to \(\textsf{rest}\), \(\textsf{walk}\), and \(\textsf{run}\). Within each fiber, all transitions are labeled by the fiber label, so the visible symbol is determined by the destination side of the move.

The phase index records the remaining forced-hold depth before a non-self departure becomes possible. For example, after entering the \(\textsf{walk}\) fiber, the controller first visits phase \(1\), and one further occurrence of \(\textsf{walk}\) is required before it reaches the decision state at phase \(0\). Only there can it either remain in \(\textsf{walk}\) or move to another fiber.

This is the shape that will later be characterized intrinsically: path-like fibers with one decision state and destination-labeled transitions between fibers.

The phase-expanded picture in \autoref{fig:running_example_dlsl} is enough for the rest of the main line of the paper. A complementary behavioral description in terms of visible-state traces and run-length constraints is still useful, but it is secondary to the structural classification and is therefore deferred to \autoref{app:runlength}.

\section{Exact graph-preserving state cost}\label{sec:exact_state_cost}

The phase-expanded realization uses extra control states to store the residual dwell phase. The next question is whether this is merely a convenient construction or the exact graph-preserving state cost of dwell.

\subsection{Graph-respecting transducers}

\begin{defi}[Graph-respecting labeled transducer]
\label{def:graph_respecting}
A \emph{graph-respecting labeled transducer} for the visible graph \((S,E)\) is a deterministic transducer
\[
T=(Q,q_0,\Omega,S,\Delta,\lambda),
\]
where \(\lambda:Q\to S\) is a surjective visible-state labeling such that for every control state \(q\in Q\) and every input symbol \(u\in\Omega\),
\[
\lambda(\Delta(q,u)) \in \Out(\lambda(q)).
\]
Its emitted symbol after reading one input is the visible label of the successor control state. We say that \(T\) is \emph{equivalent} to \(\mathcal A\) if for every input word \(w\), the emitted output word of \(T\) equals \(\Out_{\mathcal A}(w)\).
\end{defi}

\subsection{Why the visible graph alone is insufficient}

\begin{prop}[Visible states alone do not suffice in the presence of nontrivial dwell]
\label{thm:no_exact_simulation_without_augmentation}
Assume there exists a reachable visible state \(i\in S\) such that
\[
d_i > 1
\qquad \text{and} \qquad
\Out(i) \setminus \{i\} \ne \varnothing.
\]
Assume further that at least one non-self destination from \(i\) is realizable, namely that there exist \(j \in \Out(i)\setminus\{i\}\) and \(b \in \Omega\) such that
\[
g_i(b)=\ell(j).
\]
Then there is no deterministic graph-respecting labeled transducer equivalent to \(\mathcal A\) whose control-state set is exactly \(S\) and whose visible-state labeling is the identity.
\end{prop}

\begin{proof}
The idea is to compare two histories that end at the same visible state \(i\) but leave different amounts of dwell budget.

Because \(i\) is reachable, there exists an input word \(\alpha\) that brings the symbolic DLSL system to visible state \(i\). By extending \(\alpha\), we may assume that the operational state after reading \(\alpha\) is \((i,d_i-1)\), namely the state immediately after entry into \(i\). Fix an arbitrary input symbol \(c \in \Omega\).

Consider the two words
\[
w = \alpha c^{d_i-1},
\qquad
w' = \alpha c^{d_i-2}.
\]
After reading \(w\), the operational state is \((i,0)\): the machine is still visibly in \(i\), but it is now allowed to leave. After reading \(w'\), the operational state is \((i,1)\): the machine is still visibly in \(i\), but one forced-hold step remains.

Now feed the same additional symbol \(b\). By assumption, when the machine is in visible state \(i\) and residual phase \(0\), the symbol \(b\) triggers a move to \(j \ne i\). Hence the visible output on the continuation \(b\) after history \(w\) begins with \(j\). By contrast, after history \(w'\), the machine is still under forced hold. Here one uses \eqref{eq:operational_update} explicitly: for every input symbol \(u\), if the residual phase is positive then the next operational state remains in the same visible state, namely \((i,r)\xrightarrow{u}(i,r-1)\). Therefore the same continuation \(b\) keeps the visible state at \(i\); the visible output begins with \(i\). The two continuations thus produce different future visible outputs.

Suppose now that an equivalent graph-respecting transducer existed with control-state set exactly \(S\) and visible-state labeling equal to the identity. After histories \(w\) and \(w'\), that transducer would be in the same control state \(i\). By determinism, feeding the same next symbol \(b\) would then produce the same successor visible state after both histories, contradicting the previous paragraph. Thus no such transducer exists. \qedhere
\end{proof}

In the running example, the same phenomenon appears inside the visible state \(\mathrm{walk}\): two histories may both end in \(\mathrm{walk}\), yet only the one that has already discharged its dwell budget can respond immediately to a proposal for \(\mathrm{run}\). That distinction is invisible on the visible graph alone.

\subsection{A lower bound above each visible state}

\begin{thm}[Lower bound on required augmentation]
\label{thm:lower_bound_augmentation}
Assume the symbolic DLSL system is destination-complete. Let
\[
T=(Q,q_0,\Omega,S,\Delta,\lambda)
\]
be any deterministic graph-respecting labeled transducer equivalent to \(\mathcal A\).

Fix a reachable visible state \(i\in S\) such that \(\Out(i)\setminus\{i\} \ne \varnothing\) and at least one non-self departure from \(i\) is realizable. Then
\[
|\lambda^{-1}(i)| \ge d_i.
\]
Consequently, if every visible state is reachable and admits a realizable non-self departure, then
\[
|Q| \ge \sum_{i\in S} d_i.
\]
\end{thm}

\begin{proof}
This is the same separation idea, now used to distinguish all residual phases above a fixed visible state.

Write \(m=d_i\). We show that the \(m\) residual phases above visible state \(i\) are pairwise distinguishable by future continuations. Let \(0 \le r < r' \le m-1\). Choose a non-self realizable departure symbol \(b \in \Omega\) and destination \(j \in \Out(i)\setminus\{i\}\) such that
\[
g_i(b)=\ell(j).
\]
Fix an arbitrary symbol \(c \in \Omega\), and consider the continuation
\[
x = c^r b.
\]

Start from the operational state \((i,r)\). The prefix \(c^r\) consumes exactly the remaining forced-hold budget and brings the machine to phase \(0\). The final symbol \(b\) is therefore read at decision phase, so the machine departs from \(i\) to \(j\). Thus the output on continuation \(x\) begins with a block of \(i\)'s of length exactly \(r\), followed by \(j\).

Now start instead from operational state \((i,r')\). After the same prefix \(c^r\), the residual phase is still \(r'-r>0\), so the final symbol \(b\) is read under forced hold. Hence the machine remains in visible state \(i\) throughout the whole continuation \(x\). In particular, the resulting output word differs from the previous one.

Therefore the future behaviors obtained from residual phases \(r\) and \(r'\) are different. Any equivalent deterministic transducer must realize these distinguishable behaviors by distinct reachable control states. Since all these behaviors project to the same visible state \(i\), the fiber \(\lambda^{-1}(i)\) must contain at least \(m=d_i\) distinct control states.

If the assumption holds for every visible state, the fibers \(\lambda^{-1}(i)\) are disjoint and summing the per-fiber lower bounds yields
\[
|Q| = \sum_{i\in S} |\lambda^{-1}(i)|
    \ge \sum_{i\in S} d_i.
\]
\end{proof}

\begin{cor}[Tightness of the phase expansion]
\label{cor:tightness}
Under the hypotheses of \autoref{thm:lower_bound_augmentation}, the phase-expanded realization of \autoref{thm:compilation_transducer} is optimal among deterministic graph-respecting labeled transducers equivalent to \(\mathcal A\). In particular,
\[
\min |Q| = \sum_{i\in S} d_i.
\]
\end{cor}

\begin{proof}
\autoref{thm:compilation_transducer} constructs an equivalent graph-respecting deterministic realization with exactly \(\sum_i d_i\) control states, while \autoref{thm:lower_bound_augmentation} shows that no equivalent graph-respecting deterministic realization can use fewer. The bounds match. \qedhere
\end{proof}

So the dwell budget above each visible state has an exact deterministic state cost. The phase-expanded realization is not just sufficient; in the graph-preserving setting it is also minimal.

\section{Intrinsic characterization: path fibers, chain monoids, and converse representation}\label{sec:intrinsic_characterization}

We now identify the class of deterministic transducers that arise from phase expansion, without assuming phase coordinates in advance. The strategy is to define an intrinsic structural property --- fiber-linearity --- purely in terms of the hold behavior inside each visible-state fiber and the reset behavior of decisions. We then show that this property characterizes exactly the compiled DLSL realizations and that it yields a unique canonical representative.

\subsection{Fiber-linear transducers}

Recall that a graph-respecting labeled transducer is a deterministic transducer
\[
T=(Q,q_0,\Omega,S,\Delta,\lambda),
\]
where \(\lambda:Q\to S\) is a surjective visible-state labeling such that
\[
\lambda(\Delta(q,u))\in\Out(\lambda(q))
\qquad \text{for all } q\in Q,\ u\in\Omega.
\]

\begin{defi}[Forced-hold state]
\label{def:forced_hold}
Let \(T=(Q,q_0,\Omega,S,\Delta,\lambda)\) be a graph-respecting transducer. A control state \(q\in Q\) is a \emph{forced-hold state} if there exists a control state \(q'\in Q\) such that
\[
\lambda(q')=\lambda(q)
\qquad\text{and}\qquad
\Delta(q,u)=q'
\quad\text{for every }u\in\Omega.
\]
In that case, \(q'\) is uniquely determined by determinism, and we denote it by
\[
\holdsucc(q).
\]
A control state that is not forced-hold is called a \emph{decision state}.
\end{defi}

\begin{defi}[Fiber hold graph]
\label{def:fiber_hold_graph}
For a visible state \(i\in S\), let
\[
\Fib(i):=\lambda^{-1}(i).
\]
The \emph{fiber hold graph} \(H_i\) is the directed graph with vertex set \(\Fib(i)\) and an edge
\[
q \longrightarrow \holdsucc(q)
\]
for every forced-hold state \(q\in\Fib(i)\).
\end{defi}

\begin{defi}[Fiber-linear graph-respecting transducer]
\label{def:fiber_linear}
A graph-respecting transducer
\[
T=(Q,q_0,\Omega,S,\Delta,\lambda)
\]
is \emph{fiber-linear} if, for every visible state \(i\in S\), the following hold:
\begin{enumerate}
    \item \textbf{Path condition inside the fiber:} the fiber hold graph \(H_i\) is a directed simple path on the whole fiber \(\Fib(i)\), ending at a unique terminal vertex \(p_i\). Equivalently, \(p_i\) is the unique decision state in \(\Fib(i)\), every other state in \(\Fib(i)\) is forced-hold, and repeated application of \(\holdsucc\) walks through the fiber in a single linear chain until \(p_i\) is reached.
    \item \textbf{Reset-on-decision:} let \(t_i\) denote the unique initial vertex of the path \(H_i\) (so \(t_i=p_i\) when \(|\Fib(i)|=1\)). Then for every input symbol \(u\in\Omega\), if
    \[
    \lambda(\Delta(p_i,u)) = j,
    \]
    we have
    \[
    \Delta(p_i,u)=t_j.
    \]
\end{enumerate}
\end{defi}

The point of \autoref{def:fiber_linear} is that it is intrinsic. It refers only to input-independent hold behavior inside visible-state fibers and to where decisions reset. No phase indexing is assumed at the outset.

\autoref{fig:fiber_linear_structure} makes the geometry explicit. The positive panel shows the unique pattern allowed inside each visible-state fiber: a linear hold chain culminating in one decision state, with every inter-fiber decision entering the designated entry state of the successor fiber. The two obstruction panels anticipate \autoref{cor:obstruction}: any internal cycle or any decision entering an interior state of a successor fiber immediately falls outside the DLSL class.

% ------------------------------------------------------------------
% FIGURE A: Fiber-linear structure and obstructions
% Put immediately after \autoref{def:fiber_linear}, i.e. after the paragraph
% ending with “No phase indexing is assumed at the outset.”
% ------------------------------------------------------------------
\begin{figure}[pos=H]
\centering
\begin{tikzpicture}[
    >=Latex,
    state/.style={circle,draw,minimum size=8mm,inner sep=0pt,font=\small},
    term/.style={circle,draw,thick,minimum size=8mm,inner sep=0pt,font=\small},
    bad/.style={circle,draw,dashed,minimum size=8mm,inner sep=0pt,font=\small},
    lab/.style={font=\small}
]

% Main good panel
\node[font=\small\bfseries] at (0,3.0) {(a) Fiber-linear fiber structure};

% Left fiber
\node[state] (ti) at (-2.8,1.4) {$t_i$};
\node[state] (mi) at (-2.8,-0.6) {$\cdot$};
\node[term]  (pi) at (-2.8,-2.6) {$p_i$};
\draw[->,thick] (ti) -- (mi);
\draw[->,thick] (mi) -- (pi);
\node[lab] at (-2.0,0.4) {hold};
\node[lab] at (-2.0,-1.6) {hold};
\node[lab] at (-2.8,-3.4) {$\Fib(i)$};

% Right fiber
\node[state] (tj) at (2.0,1.4) {$t_j$};
\node[state] (mj) at (2.0,-0.6) {$\cdot$};
\node[term]  (pj) at (2.0,-2.6) {$p_j$};
\draw[->,thick] (tj) -- (mj);
\draw[->,thick] (mj) -- (pj);
\node[lab] at (2.8,0.4) {hold};
\node[lab] at (2.8,-1.6) {hold};
\node[lab] at (2.0,-3.4) {$\Fib(j)$};

\draw[->,very thick,bend left=12] (pi.east) to node[above,lab] {decision on $u$} (tj.west);
\node[lab,align=center] at (0,-4.4) {each fiber is a single hold path, and each decision from $p_i$ resets to the entry state $t_j$};

% Obstruction panels
\node[font=\small\bfseries] at (-3.0,-5.8) {(b) Forbidden: cycle inside a fiber};
\node[bad] (c1) at (-4.0,-7.1) {$q_0$};
\node[bad] (c2) at (-2.0,-7.1) {$q_1$};
\draw[->,thick] (c1) to[bend left=18] node[above,lab] {hold} (c2);
\draw[->,thick] (c2) to[bend left=18] node[below,lab] {hold} (c1);
\node[lab,align=center] at (-3.0,-8.2) {not a simple path};

\node[font=\small\bfseries] at (3.0,-5.8) {(c) Forbidden: nonentry target};
\node[state] (xt) at (1.4,-6.6) {$t_j$};
\node[bad]   (xm) at (1.4,-8.2) {$e$};
\node[term]  (xp) at (1.4,-10.2) {$p_j$};
\draw[->,thick] (xt) -- (xm);
\draw[->,thick] (xm) -- (xp);
\node[term]  (src) at (4.3,-8.2) {$p_i$};
\draw[->,very thick,bend left=10] (src.west) to node[above,lab] {$u$} (xm.east);
\node[lab,align=center] at (2.8,-11.4) {decision enters an interior state of $\Fib(j)$\\instead of the designated entry state};

\end{tikzpicture}
\caption{Fiber-linear structure and minimal obstructions. Panel (a) shows the geometric content of \autoref{def:fiber_linear}: each visible-state fiber is a single directed hold path ending at one decision state, and every decision transition resets to the entry state of the successor fiber. Panels (b) and (c) show two immediate obstructions used later in \autoref{cor:obstruction}.}
\label{fig:fiber_linear_structure}
\end{figure}
The same class can also be described in standard algebraic language. Each fiber carries a single unary ``countdown'' transformation whose powers form a finite chain monoid; the reset condition says that cross-fiber actions factor through the unique active state at the bottom of that chain.

\begin{prop}[Transformation-monoid characterization]
\label{prop:semigroup_characterization}
Let
\[
T=(Q,q_0,\Omega,S,\Delta,\lambda)
\]
be a graph-respecting transducer. Then the following are equivalent.
\begin{enumerate}
    \item \(T\) is fiber-linear.
    \item For every visible state \(i\in S\), there exist states \(t_i,p_i\in \Fib(i)\) and a transformation
    \[
    h_i:\Fib(i)\to \Fib(i)
    \]
    such that:
    \begin{enumerate}
        \item the transformation monoid \(\langle h_i\rangle\) has size \(|\Fib(i)|\), is generated by the single map \(h_i\), and is a finite chain monoid in the sense that its unique nonidentity idempotent is the constant map onto \(p_i\);
        \item for every \(q\in \Fib(i)\setminus\{p_i\}\) and every input symbol \(u\in\Omega\), one has
        \[
        \Delta(q,u)=h_i(q);
        \]
        in particular, every nonterminal fiber state is input-independent;
        \item if \(\lambda(\Delta(p_i,u))=j\), then
        \[
        \Delta(p_i,u)=t_j.
        \]
    \end{enumerate}
\end{enumerate}
In particular, fiber-linear transducers are exactly the graph-respecting deterministic transducers whose per-fiber dynamics are governed by monogenic aperiodic chain monoids and whose cross-fiber actions factor through the unique active state of each fiber.
\end{prop}

\begin{proof}
Assume first that \(T\) is fiber-linear. For a fiber \(\Fib(i)\), let \(p_i\) be the terminal vertex of the hold path and \(t_i\) its initial vertex. Define
\[
h_i(q):=
\begin{cases}
\holdsucc(q), & q\neq p_i,\\
p_i, & q=p_i.
\end{cases}
\]
Because the fiber hold graph is a directed simple path on all vertices of \(\Fib(i)\), repeated application of \(h_i\) moves every state one step closer to \(p_i\), and the powers
\[
\mathrm{id},h_i,h_i^2,\dots,h_i^{|\Fib(i)|-1}
\]
are pairwise distinct. The last power is the constant map onto \(p_i\), and it is the unique nonidentity idempotent. Clause (b) is just the definition of forced-hold, and clause (c) is the reset-on-decision property.

Conversely, assume the transformation-monoid formulation. Since \(\langle h_i\rangle\) is a finite chain monoid of size \(|\Fib(i)|\), the orbit of every state under \(h_i\) is linearly ordered and terminates at the unique image point \(p_i\) of the nonidentity idempotent. Because the monoid has exactly \(|\Fib(i)|\) elements, there is exactly one state at each distance from \(p_i\), so the directed graph formed by the edges \(q\to h_i(q)\) for \(q\neq p_i\) is a simple path on the whole fiber, with initial vertex \(t_i\) and terminal vertex \(p_i\). Clause (b) then says that every nonterminal fiber state is forced-hold, and clause (c) is precisely reset-on-decision. Hence \(T\) is fiber-linear.
\end{proof}

\paragraph{Position inside classical deterministic transducers.}
Every compiled DLSL realization is a sequential transducer, hence belongs to the classical subsequential world. The inclusion is strict: an accessible deterministic graph-respecting sequential transducer may fail to be fiber-linear simply because one visible-state fiber contains a cycle instead of a single hold path, or because a decision transition enters a nonentry state of a successor fiber.

\begin{thm}[Compiled realizations are fiber-linear]
\label{thm:compiled_fiber_linear}
Let \(\mathcal A\) be a symbolic DLSL system with dwell, and let \(T_{\mathcal A}\) be the deterministic realization constructed in \autoref{thm:compilation_transducer}. Then \(T_{\mathcal A}\) is fiber-linear.
\end{thm}

\begin{proof}
The phase-expanded realization has control states
\[
Q=\{(i,r): i\in S,\ 0\le r\le d_i-1\},
\]
with visible-state labeling \(\lambda(i,r)=i\).

Fix a visible state \(i\). For every state \((i,r)\) with \(r>0\), the operational update is
\[
\Delta((i,r),u)=(i,r-1)
\qquad \text{for all } u\in\Omega,
\]
so \((i,r)\) is forced-hold and
\[
\holdsucc((i,r))=(i,r-1).
\]
The unique state in the fiber \(\Fib(i)\) that is not forced-hold is \((i,0)\), since from that state the successor depends on the input through \(g_i\). Therefore the fiber hold graph is exactly the directed path
\[
(i,d_i-1)\to(i,d_i-2)\to\cdots\to(i,1)\to(i,0).
\]
Thus the path condition holds, with terminal vertex \(p_i=(i,0)\) and initial vertex \(t_i=(i,d_i-1)\).

Now consider an input symbol \(u\in\Omega\) read at the decision state \((i,0)\). By the operational semantics,
\[
\Delta((i,0),u)=\bigl(j,d_j-1\bigr),
\qquad
j=\delta(i,g_i(u)).
\]
But \((j,d_j-1)=t_j\), the initial vertex of the path in the fiber over \(j\). Hence the reset-on-decision condition also holds. Therefore \(T_{\mathcal A}\) is fiber-linear. \qedhere
\end{proof}

Theorem~\ref{thm:compiled_fiber_linear} shows that every compiled DLSL realization is fiber-linear. The next theorem establishes the converse: fiber-linearity is not merely a property of compiled realizations but characterizes them exactly.

\begin{thm}[Intrinsic converse representation]
\label{thm:intrinsic_converse}
Let
\[
T=(Q,q_0,\Omega,S,\Delta,\lambda)
\]
be a fiber-linear graph-respecting transducer. Then there exists an injectively labeled symbolic DLSL system with dwell
\[
\mathcal A_T=(S,E,\ell,s_0,d,(g_i)_{i\in S})
\]
such that the phase-expanded realization of \(\mathcal A_T\) is isomorphic to \(T\). Moreover, \(\mathcal A_T\) is unique up to injective relabeling of the visible states.
\end{thm}

\begin{proof}
For each visible state \(i\), let \(p_i\) be the unique terminal vertex of the path \(H_i\), and let \(t_i\) be the unique initial vertex of that path. Define
\[
d_i := |\Fib(i)|.
\]
Because \(H_i\) is a directed simple path on \(\Fib(i)\), each state \(q\in\Fib(i)\) has a unique distance to the terminal vertex \(p_i\). Write this distance as
\[
\Phase_i(q)\in\{0,1,\dots,d_i-1\},
\]
so that
\[
\Phase_i(p_i)=0,
\qquad
\Phase_i(t_i)=d_i-1.
\]

We now reconstruct a symbolic DLSL system. Let the visible alphabet be the visible-state set itself:
\[
\Sigma:=S,
\qquad
\ell(i):=i.
\]
This labeling is injective. Define the visible edge set by
\[
E :=
\{(i,j): \exists u\in\Omega \text{ such that } \lambda(\Delta(p_i,u))=j\}
\;\cup\;
\{(i,i): i\in S\}.
\]
The added self-loops make the visible graph self-looping. Because the labeling is injective, the graph is automatically destination-labeled and label-deterministic.

For each visible state \(i\), define
\[
g_i(u):=\lambda(\Delta(p_i,u)).
\]
Since \(\lambda(\Delta(p_i,u))\in\Out(i)\), we indeed have \(g_i(u)\in\Sigma_i\).

Consider now the phase-expanded realization of this symbolic DLSL system. Its control states are pairs
\[
(i,r),\qquad 0\le r\le d_i-1.
\]
Define
\[
\Phi:Q\to \{(i,r): i\in S,\ 0\le r\le d_i-1\}
\]
by
\[
\Phi(q):=\bigl(\lambda(q),\Phase_{\lambda(q)}(q)\bigr).
\]
Because each fiber hold graph is a path, every fiber contains exactly one state at each phase depth \(r\), so \(\Phi\) is bijective.

We show that \(\Phi\) is a transducer isomorphism.

First, suppose \(q\in\Fib(i)\) is forced-hold. Then \(\holdsucc(q)\) is the next state on the path toward \(p_i\), hence
\[
\Phase_i(\holdsucc(q))=\Phase_i(q)-1.
\]
Since \(q\) is forced-hold, \(\Delta(q,u)=\holdsucc(q)\) for every \(u\in\Omega\). Therefore
\[
\Phi(\Delta(q,u))
=
\bigl(i,\Phase_i(q)-1\bigr),
\]
which is exactly the phase-expanded transition from \((i,\Phase_i(q))\) when the phase is positive.

Second, suppose \(q=p_i\) is the decision state in fiber \(i\). By fiber-linearity, for every \(u\in\Omega\), if
\[
\lambda(\Delta(p_i,u))=j,
\]
then
\[
\Delta(p_i,u)=t_j.
\]
Hence
\[
\Phi(\Delta(p_i,u))
=
\bigl(j,\Phase_j(t_j)\bigr)
=
(j,d_j-1),
\]
which is exactly the phase-expanded reset transition from decision phase \(0\) in the symbolic DLSL system defined above.

Thus \(\Phi\) intertwines the transitions of \(T\) and of the phase-expanded realization of \(\mathcal A_T\), and it preserves visible labels by construction. Therefore the two transducers are isomorphic.

Finally, the reconstruction depends only on the visible graph, the fiber lengths, and the injective naming of visible labels. Replacing \(\ell(i)=i\) by any other injective relabeling yields an isomorphic symbolic DLSL system. This is the claimed uniqueness up to injective relabeling. \qedhere
\end{proof}

\autoref{fig:reconstruction_canonicality} summarizes the logic of the converse theorem. Fiber inspection extracts the visible graph, the dwell vector, and the local decision maps; those data in turn determine a unique phase-expanded realization. This reconstruction viewpoint is what later supports both the obstruction test and the canonicality result.

% ------------------------------------------------------------------
% FIGURE B: Reconstruction and canonicality pipeline
% Put immediately after the proof of \autoref{thm:intrinsic_converse}, i.e.
% before \autoref{cor:obstruction}. It can also go immediately before
% \autoref{thm:behavioral_rigidity} if preferred.
% ------------------------------------------------------------------
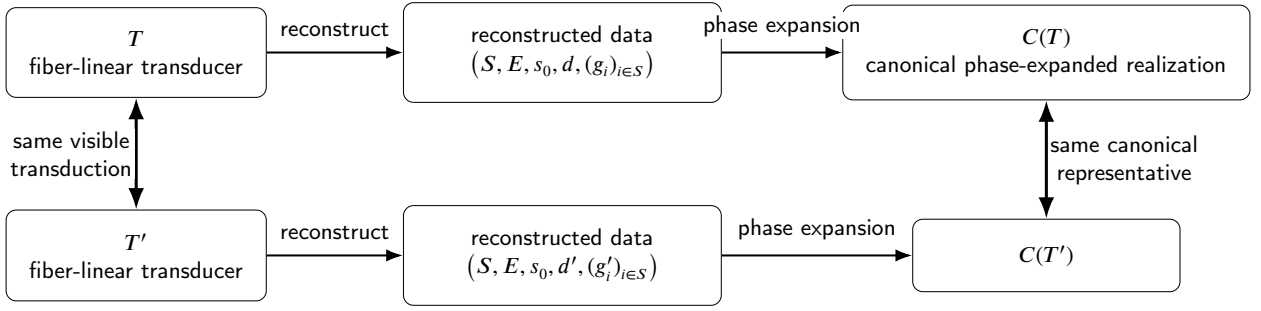
\begin{figure}[pos=H]
\centering
\resizebox{\linewidth}{!}{%
\begin{tikzpicture}[
    >=Latex,
    box/.style={draw,rounded corners,align=center,minimum height=10mm,inner sep=3.2mm,font=\small},
    arr/.style={->,thick},
    edgelab/.style={font=\small,fill=white,inner xsep=2pt,inner ysep=1pt}
]
\node[box,minimum width=30mm] (T)  at (0,0)      {$T$\\fiber-linear transducer};
\node[box,minimum width=44mm] (R)  at (5.9,0)    {reconstructed data\\$\bigl(S,E,s_0,d,(g_i)_{i\in S}\bigr)$};
\node[box,minimum width=37mm] (C)  at (12.6,0)   {$C(T)$\\canonical phase-expanded realization};
\draw[arr] (T) -- node[midway,above=5pt,edgelab] {reconstruct} (R);
\draw[arr] (R) -- node[midway,above=5pt,edgelab] {phase expansion} (C);

\node[box,minimum width=30mm] (Tp) at (0,-2.8)   {$T'$\\fiber-linear transducer};
\node[box,minimum width=44mm] (Rp) at (5.9,-2.8) {reconstructed data\\$\bigl(S,E,s_0,d',(g'_i)_{i\in S}\bigr)$};
\node[box,minimum width=37mm] (Cp) at (12.6,-2.8) {$C(T')$};
\draw[arr] (Tp) -- node[midway,above=5pt,edgelab] {reconstruct} (Rp);
\draw[arr] (Rp) -- node[midway,above=5pt,edgelab] {phase expansion} (Cp);

\draw[<->,very thick] (T.south) -- node[left,font=\small,align=center,fill=white,inner sep=1pt] {same visible\\transduction} (Tp.north);
\draw[<->,very thick] (C.south) -- node[right,font=\small,align=center,fill=white,inner sep=1pt] {same canonical\\representative} (Cp.north);
\end{tikzpicture}%
}
\caption{Reconstruction and canonical representative. A fiber-linear transducer determines its dwell vector and local decision maps by fiber inspection, and these data determine a unique phase-expanded realization. This is the mechanism behind \autoref{thm:intrinsic_converse}, \autoref{thm:behavioral_rigidity}, and \autoref{cor:equivalence_canonical}.}
\label{fig:reconstruction_canonicality}
\end{figure}
\begin{cor}[Simple obstructions to DLSL representability]
\label{cor:obstruction}
Let
\[
T=(Q,q_0,\Omega,S,\Delta,\lambda)
\]
be a deterministic graph-respecting transducer. If, for some visible state \(i\in S\), the fiber \(\Fib(i)\) fails to form a single directed hold path with a unique decision state, or if some decision transition enters a noninitial state of a successor fiber, then \(T\) is not isomorphic to the phase-expanded realization of any injectively labeled symbolic DLSL system with dwell.
\end{cor}

\begin{proof}
Any such transducer is not fiber-linear. The claim is therefore immediate from \autoref{thm:intrinsic_converse}.
\end{proof}

\begin{thm}[Behavioral rigidity and canonicality]
\label{thm:behavioral_rigidity}
Let
\[
T=(Q,q_0,\Omega,S,\Delta,\lambda)
\qquad\text{and}\qquad
T'=(Q',q'_0,\Omega,S,\Delta',\lambda')
\]
be accessible fiber-linear graph-respecting transducers over the same visible graph \((S,E)\) and with the same initial visible state. Assume that every visible state is reachable and admits a realizable non-self departure in both transducers. If \(T\) and \(T'\) induce the same visible-state transduction
\[
\Out_T(w)=\Out_{T'}(w)
\qquad \text{for all } w\in\Omega^*,
\]
then \(T\) and \(T'\) are isomorphic.
\end{thm}

\begin{proof}
By \autoref{thm:intrinsic_converse}, both transducers reconstruct to symbolic DLSL systems
\[
\mathcal A=(S,E,s_0,d,(g_i)_{i\in S})
\qquad\text{and}\qquad
\mathcal A'=(S,E,s_0,d',(g'_i)_{i\in S}),
\]
where we suppress the injective labeling because only the visible-state behavior matters here. The behavioral argument below is carried out on the reconstructed DLSL systems \(\mathcal A\) and \(\mathcal A'\), not directly on \(T\) and \(T'\). This is valid because \autoref{thm:intrinsic_converse} gives an isomorphism between each fiber-linear transducer and the phase-expanded realization of its reconstructed data; hence visible transduction equality for \(T\) and \(T'\) implies visible transduction equality for \(\mathcal A\) and \(\mathcal A'\).

Fix a visible state \(i\in S\). Since \(i\) is reachable, choose a word \(x_i\) whose visible-state run enters \(i\) at its last step; if \(i=s_0\), take \(x_i=\varepsilon\). In a phase-expanded DLSL realization, reading \(x_i\) places the machine at the top state of the fiber over \(i\). Because \(i\) admits a realizable non-self departure, the set
\[
D_i:=\{m\ge 1 : \exists y\in\Omega^m \text{ such that the visible run after } x_i y \text{ leaves } i \text{ at its last step}\}
\]
is nonempty. Its minimum is exactly \(d_i\): departure is impossible during the first \(d_i-1\) symbols after entry into \(i\), while some departure is realizable on the \(d_i\)-th symbol. The same reasoning applied to \(T'\) yields \(\min D_i=d'_i\). Since \(T\) and \(T'\) have the same visible transduction, the sets \(D_i\) coincide, and therefore
\[
d_i=d'_i
\qquad\text{for all } i\in S.
\]

Now fix \(u\in\Omega\). Because the first \(d_i-1\) steps after entering \(i\) are forced holds, the visible successor chosen on the \(d_i\)-th step depends only on \(u\). More precisely, for any word \(z\in\Omega^{d_i-1}\), the visible state reached after reading \(x_i z u\) is
\[
\delta(i,g_i(u))
\]
in \(T\) and
\[
\delta(i,g'_i(u))
\]
in \(T'\). Equality of the visible transductions implies
\[
\delta(i,g_i(u))=\delta(i,g'_i(u)).
\]
Since the visible graph is fixed, the destination visible state uniquely determines the local choice. Hence \(g_i(u)=g'_i(u)\) for every \(u\in\Omega\), and therefore
\[
g_i=g'_i
\qquad\text{for all } i\in S.
\]

Thus the reconstructed DLSL data coincide. Their phase-expanded realizations are therefore identical up to the canonical phase-state renaming, which yields an isomorphism between \(T\) and \(T'\).
\end{proof}

\begin{cor}[Canonical representative and equivalence test]
\label{cor:equivalence_canonical}
Every accessible fiber-linear graph-respecting transducer has a unique canonical representative, namely the phase-expanded realization of its reconstructed DLSL data. In particular, equivalence of two accessible fiber-linear transducers over the same visible graph is decidable in polynomial time by reconstruction and comparison of the canonical data.
\end{cor}

\begin{proof}
The canonical representative is provided by \autoref{thm:intrinsic_converse}. Uniqueness follows from \autoref{thm:behavioral_rigidity}, and polynomial-time reconstruction from \autoref{thm:recognition}.\qedhere
\end{proof}

\begin{rem}
Concretely, two fiber-linear transducers realizing the same visible-state transduction can be tested for equivalence by reconstructing their dwell vectors \((d_i)_{i\in S}\) and local maps \((g_i)_{i\in S}\) using the algorithm of \autoref{thm:recognition}, and comparing the resulting DLSL data directly. No state-space exploration, product construction, or language-equivalence test is needed. In particular, the canonical representative is computed in \(O(|Q||\Omega|)\) time.
\end{rem}

\section{Recognition and reconstruction}\label{sec:recognition}

The converse characterization is constructive. It turns DLSL representability into a directly testable property of a deterministic transducer description, and the defining fiber structure can be checked algorithmically.

\begin{thm}[Polynomial-time recognition and reconstruction]
\label{thm:recognition}
There is an algorithm that, given a deterministic graph-respecting transducer
\[
T=(Q,q_0,\Omega,S,\Delta,\lambda),
\]
decides in time \(O(|Q||\Omega|)\) whether \(T\) is fiber-linear. When the answer is yes, the algorithm reconstructs the unique underlying symbolic DLSL system of \autoref{thm:intrinsic_converse}, up to injective relabeling of the visible states.
\end{thm}

\begin{proof}
The test is straightforward: identify forced-hold states, inspect each fiber, and then check whether decisions reset to the top of successor fibers.

For each control state \(q\in Q\), inspect the set of successors
\[
\{\Delta(q,u): u\in\Omega\}.
\]
If all these successors coincide with a state \(q'\) satisfying
\[
\lambda(q')=\lambda(q),
\]
then mark \(q\) as forced-hold and set \(\holdsucc(q)=q'\); otherwise mark \(q\) as a decision state. This takes \(O(|Q||\Omega|)\) time.

Now fix a visible state \(i\in S\). Build the fiber hold graph \(H_i\) on \(\Fib(i)=\lambda^{-1}(i)\), with edges \(q\to\holdsucc(q)\) for forced-hold states \(q\in\Fib(i)\). We claim that \(T\) is fiber-linear if and only if, for every \(i\), the following conditions hold:
\begin{enumerate}
    \item \(H_i\) is a directed simple path on all vertices of \(\Fib(i)\);
    \item the unique terminal vertex of \(H_i\) is the unique decision state \(p_i\) in \(\Fib(i)\);
    \item if \(t_i\) denotes the unique initial vertex of \(H_i\), then for every \(u\in\Omega\),
    \[
    \Delta(p_i,u)=t_j
    \quad\text{whenever}\quad
    \lambda(\Delta(p_i,u))=j.
    \]
\end{enumerate}
These are exactly the two clauses of \autoref{def:fiber_linear}, rewritten in graph-theoretic language.

Each of the three checks above can be carried out in time linear in the size of the fiber plus \(|\Omega|\). For the path condition, one does not merely test connectivity and acyclicity. One checks directly that, within each fiber \(H_i\), every forced-hold vertex except one has indegree \(1\) and outdegree \(1\), the unique initial vertex has indegree \(0\), the terminal vertex \(p_i\) has outdegree \(0\), and all vertices belong to the same weakly connected component. These conditions are exactly the directed simple-path conditions for the hold graph on the fiber. Summing over all fibers yields total time
\[
O(|Q|+|Q||\Omega|)=O(|Q||\Omega|).
\]

If every fiber passes, then the path order in each fiber is uniquely determined. The length of the fiber becomes the dwell value \(d_i\), the terminal state becomes the decision state \(p_i\), the initial state becomes the top state \(t_i\), and the reconstructed local maps are
\[
g_i(u):=\lambda(\Delta(p_i,u)).
\]
This is exactly the reconstruction used in \autoref{thm:intrinsic_converse}. Hence the underlying symbolic DLSL system is recovered, uniquely up to injective relabeling.
\end{proof}

\section{Extensions beyond the canonical normal form}

The core converse theorem uses injective visible labels and exact reset-to-top because that is the cleanest canonical normal form. Neither restriction is conceptually fundamental.

\subsection{Injective labels as a presentation choice}

The intrinsic reconstruction theorem is really about visible states, fibers, and entry structure. Injective visible labels are used only to choose canonical names for the reconstructed visible states. Concretely, once a fiber-linear graph-respecting transducer has been reconstructed to its canonical injectively labeled DLSL presentation, any alternative visible-state labeling that remains destination-labeled and label-deterministic on the same visible graph yields an equivalent presentation of the same visible-state dynamics. In that sense, injective labels are a presentation choice rather than part of the intrinsic structure.

\subsection{Edge-entry DLSL systems}\label{sec:edge_entry_extension}

A second extension relaxes reset-to-top. Instead of forcing every decision from \(i\) to enter the top state of the successor fiber \(j\), we may allow the entry phase to depend on the visible edge \((i,j)\).

\begin{defi}[Edge-entry DLSL system]
An \emph{edge-entry DLSL system} consists of
\[
\mathcal A=(S,E,\ell,s_0,d,\rho,(g_i)_{i\in S}),
\]
where \((S,E,\ell)\) is destination-labeled, self-looping, and label-deterministic, where \(d_i\ge 1\), where each \(g_i:\Omega\to\Sigma_i\) is a local decision map, and where
\[
\rho:E\to\mathbb N
\]
satisfies \(0\le \rho(i,j)\le d_j-1\) for every edge \((i,j)\in E\). The operational update becomes
\[
(i,r) \xrightarrow{u}
\begin{cases}
(i,r-1), & r>0,\\
(j,\rho(i,j)), & r=0 \text{ and } j=\delta(i,g_i(u)).
\end{cases}
\]
The original DLSL model is the special case \(\rho(i,j)=d_j-1\) for every edge.
\end{defi}

\autoref{fig:edge_entry_reset} shows exactly what changes in this extension. The internal geometry of each fiber remains a directed hold path, but the target of a visible decision no longer has to be the top state of the successor fiber. What matters instead is that each visible edge \((i,j)\) determines a designated entry state inside \(\Fib(j)\), which is the structural condition captured next by entry-consistency.

% ------------------------------------------------------------------
% FIGURE C: Edge-entry extension
% Put immediately after the definition of edge-entry DLSL systems, i.e.
% after the paragraph ending “The original DLSL model is the special case ...”
% and before the definition of entry-consistent path-fiber transducer.
% ------------------------------------------------------------------
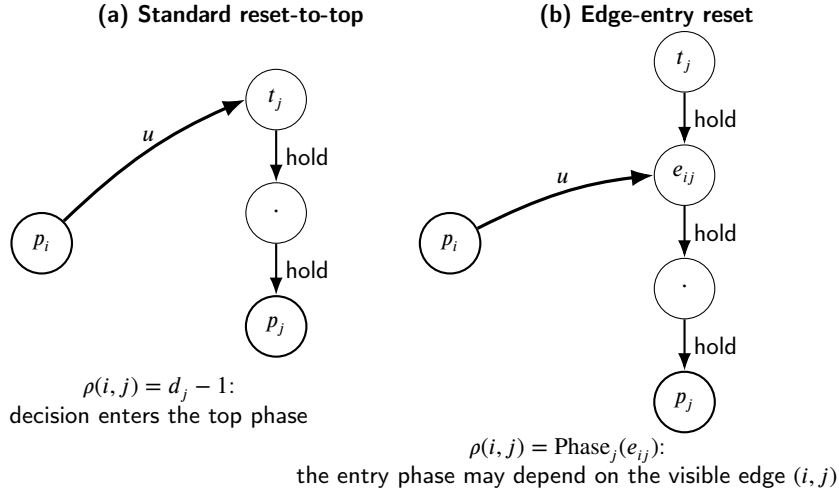
\begin{figure}[pos=H]
\centering
\begin{tikzpicture}[
    >=Latex,
    node distance=9mm and 11mm,
    state/.style={circle,draw,minimum size=8mm,inner sep=0pt,font=\small},
    term/.style={circle,draw,thick,minimum size=8mm,inner sep=0pt,font=\small},
    lab/.style={font=\small}
]
% Left panel: reset-to-top
\node[font=\small\bfseries] at (-2.5,2.3) {(a) Standard reset-to-top};
\node[term] (pi1) at (-5.0,-0.7) {$p_i$};
\node[state] (tj1) at (-1.9,1.2) {$t_j$};
\node[state] (mj1) at (-1.9,-0.3) {$\cdot$};
\node[term]  (pj1) at (-1.9,-1.8) {$p_j$};
\draw[->,thick] (tj1) -- node[right,lab] {hold} (mj1);
\draw[->,thick] (mj1) -- node[right,lab] {hold} (pj1);
\draw[->,very thick,bend left=10] (pi1) to node[above,lab] {$u$} (tj1.west);
\node[lab,align=center] at (-3.45,-2.8) {$\rho(i,j)=d_j-1$:\\decision enters the top phase};

% Right panel: edge-entry
\node[font=\small\bfseries] at (3.0,2.3) {(b) Edge-entry reset};
\node[term] (pi2) at (0.4,-0.7) {$p_i$};
\node[state] (tj2) at (3.5,1.7) {$t_j$};
\node[state] (ej2) at (3.5,0.2) {$e_{ij}$};
\node[state] (mj2) at (3.5,-1.3) {$\cdot$};
\node[term]  (pj2) at (3.5,-2.8) {$p_j$};
\draw[->,thick] (tj2) -- node[right,lab] {hold} (ej2);
\draw[->,thick] (ej2) -- node[right,lab] {hold} (mj2);
\draw[->,thick] (mj2) -- node[right,lab] {hold} (pj2);
\draw[->,very thick,bend left=10] (pi2) to node[above,lab] {$u$} (ej2.west);
\node[lab,align=center] at (1.95,-3.6) {$\rho(i,j)=\Phase_j(e_{ij})$:\\the entry phase may depend on the visible edge $(i,j)$};
\end{tikzpicture}
\caption{Reset-to-top versus edge-entry semantics. The original DLSL model always resets a decision from fiber $i$ into the top state of the successor fiber $j$. The edge-entry extension of \autoref{thm:edge_entry_converse} allows the target state to be a designated interior entry state $e_{ij}$ depending on the visible edge.}
\label{fig:edge_entry_reset}
\end{figure}

\begin{defi}[Entry-consistent path-fiber transducer]
A graph-respecting transducer is \emph{entry-consistent path-fiber} if every visible-state fiber is a directed simple hold path with a unique decision state and, for every visible edge \((i,j)\), there exists a distinguished entry state \(e_{ij}\in\Fib(j)\) such that every decision transition from the decision state of fiber \(i\) into visible state \(j\) lands in \(e_{ij}\).
\end{defi}

\begin{thm}[Converse theorem for edge-entry systems]
\label{thm:edge_entry_converse}
The phase-expanded realizations of edge-entry DLSL systems are exactly the entry-consistent path-fiber graph-respecting transducers. Moreover, the edge-entry parameters \(d_i\) and \(\rho(i,j)\) are reconstructed uniquely from the fiber lengths and designated entry states.
\end{thm}

\begin{proof}
The forward direction is immediate from the operational semantics: each fiber remains a directed hold path, and a decision from \(i\) to \(j\) lands in the designated state \((j,\rho(i,j))\). For the converse, let \(T\) be entry-consistent path-fiber. Reconstruct the visible graph, dwell values, and decision states exactly as in \autoref{thm:intrinsic_converse}. For each visible edge \((i,j)\), let \(e_{ij}\) be the designated entry state of \(\Fib(j)\), and define
\[
\rho(i,j):=\Phase_j(e_{ij}).
\]
The same bijection \(\Phi(q)=(\lambda(q),\Phase_{\lambda(q)}(q))\) used in \autoref{thm:intrinsic_converse} remains well-defined because each fiber is still a directed simple hold path, so \(\Phase_{\lambda(q)}(q)\) is the unique distance from \(q\) to the terminal decision state of its fiber, independent of the entry point. The operational semantics of the reconstructed edge-entry system match the transitions of \(T\) under \(\Phi\): for a decision from visible state \(i\) to visible state \(j\), the operational update lands in \((j,\rho(i,j))\), while entry-consistency ensures that \(T\) lands in the distinguished state \(e_{ij}\in\Fib(j)\) with \(\Phase_j(e_{ij})=\rho(i,j)\). Thus \(\Phi\) intertwines the two dynamics exactly. Uniqueness of the parameters follows because the path order in each fiber is unique and the designated entry states are part of the transducer itself.\qedhere
\end{proof}

\begin{cor}[Recognition of the edge-entry extension]
Accessible entry-consistent path-fiber transducers can be recognized and reconstructed in polynomial time by the same fiber inspection used in \autoref{thm:recognition}, augmented with the check that all decision transitions from a fixed source fiber into a fixed successor fiber land in a unique designated entry state.
\end{cor}

\begin{proof}
The path test is the same as before. Once the unique decision state \(p_i\) in each fiber is known, inspect all transitions \(\Delta(p_i,u)\) such that \(\lambda(\Delta(p_i,u))=j\). Entry-consistency requires that all such transitions land in the same state \(e_{ij}\in\Fib(j)\). If two such transitions land in distinct states of \(\Fib(j)\), the transducer is not entry-consistent. This check is performed by a single scan of the decision transitions, hence in time linear in \(|Q||\Omega|\).\qedhere
\end{proof}

\section{Related work}

The closest background is classical deterministic automata and transducer theory, especially sequential and subsequential transducers, minimization, and algebraic structure \cite{HopcroftMotwaniUllman2006,Sakarovitch2009,Choffrut1979,Choffrut2003,Mohri2000}. The present paper does not propose a new general model of regular functions. Its contribution is structural: it isolates a rigid subclass of graph-respecting deterministic transducers and characterizes that subclass intrinsically.

There is also a natural relation with the broader theory of regular word functions and their machine models \cite{FiliotReynier2016}. That literature studies equivalences between transducer formalisms, logical definability, and algebraic characterizations of regular functions. The present setting is narrower: the visible graph is fixed, graph-preserving realization is required, and the main outcome is a canonical internal decomposition. From that viewpoint, fiber-linearity should be read as a recognition criterion for a specific transducer subclass rather than as a general expressiveness result.

The DLSL framework is also motivated by the limitations of global sequence classifiers for systems with per-state structure. Hidden semi-Markov models \cite{Yu2010} encode dwell probabilistically and generatively but do not address deterministic realization, exact state cost, or hard graph constraints. Semi-Markov conditional random fields \cite{SarawagiCohen2004} are discriminative but globally trained, with no per-state classifier structure and no hard transition constraints. LSTM-based sequence classifiers \cite{HochreiterSchmidhuber1997} likewise learn transition structure implicitly from data. When transition graphs and minimum dwell values are physically known in advance, encoding them as hard structural constraints yields a more interpretable and structurally grounded model. The present paper provides the formal foundation for that encoding: it identifies the exact transducer class induced by those constraints, proves that the class is rigid, and gives a polynomial-time recognition procedure.

The timing aspect is related, but not identical, to timed automata and hybrid-system abstractions \cite{AlurDill1994,AlurEtAl1995}. Timed automata work over dense time and are analyzed through clocks, regions, or zones, whereas the present model uses discrete symbolic time and finite phase counters attached to visible states. The resulting questions are correspondingly different: exact deterministic realization on a fixed visible graph, exact state cost, canonicality, and recognition of the induced subclass.

There is also a conceptual connection with dwell-time conditions in switched and hybrid control \cite{Morse1996,HespanhaMorse1999,HespanhaLiberzonMorse2003,LiberzonMorse1999,Liberzon2003}. In that literature, dwell is imposed to control switching behavior of continuous dynamics. Here dwell is purely symbolic: it constrains when visible departures may occur and forces a specific finite-memory architecture inside a deterministic transducer.

\section{Conclusion}

We identified an exact structural subclass inside deterministic transducer theory. The phase-expanded realizations of DLSL systems are precisely the graph-respecting transducers whose visible-state fibers are linear hold paths with a unique decision state; equivalently, each fiber carries a monogenic aperiodic chain monoid and cross-fiber actions factor through its distinguished active state.

This characterization has three immediate consequences. First, the phase-expanded realization has exact graph-preserving state cost: enforcing dwell values \(d_i\) requires exactly \(\sum_i d_i\) control states in the deterministic graph-preserving setting. Second, the class is rigid: under natural reachability and realizable-departure hypotheses, equivalent accessible fiber-linear transducers over the same visible graph are isomorphic. Third, the same reconstruction perspective remains valid beyond the canonical normal form: injective labels are only a presentation device, and the edge-entry variant is handled by the same path-fiber analysis and admits its own converse and recognition theorem.

Two natural directions remain. One is to identify broader classes of graph-respecting deterministic transducers that admit comparable canonical decompositions. The other is to measure, in structural terms, how far a general graph-respecting transducer lies from the DLSL or edge-entry regimes.

\appendix

\section{Visible-state run-length formulation}
\label{app:runlength}

The main body of the paper is organized around the internal structure of graph-respecting deterministic transducers. There is also a complementary behavioral view, stated directly on visible-state traces. Since it is not part of the intrinsic converse route, we collect it here.

For a visible-state word \(\tau = y_0 y_1 \cdots y_n \in S^{n+1}\) and an index \(0 \le t \le n\), define the current run length by
\[
\run_{\tau}(t)
:=
\max \bigl\{ \ell' \ge 1 : t-\ell'+1 \ge 0
\text{ and }
y_{t-\ell'+1}=\cdots=y_t \bigr\}.
\]
To avoid confusion with the visible-state labeling \(\ell\), we use \(\ell'\) here only as a local dummy variable.

\begin{thm}[Characterization of realizable visible-state traces]
\label{thm:trace_characterization}
Assume that the symbolic DLSL system is destination-complete. Let
\[
\tau = y_0 y_1 \cdots y_n \in S^{n+1}
\qquad \text{with } y_0=s_0.
\]
Then there exists an input word \(w \in \Omega^n\) such that
\[
\tau = \tr_{\mathcal A}(w)
\]
if and only if the following two conditions hold:
\begin{enumerate}
    \item \textbf{Graph-feasible adjacency:}
    \[
    y_{t+1} \in \Out(y_t)
    \qquad \text{for every } t=0,\dots,n-1;
    \]
    \item \textbf{Minimum dwell before departure:} whenever \(y_t \ne y_{t+1}\),
    \[
    \run_{\tau}(t) \ge d_{y_t}.
    \]
\end{enumerate}
\end{thm}

\begin{proof}
Necessity is immediate from the operational semantics. Suppose \(\tau = \tr_{\mathcal A}(w)\) for some input word. At each time step the visible successor is either the current visible state itself, during forced hold, or a visible successor selected by the local decision map. In either case \(y_{t+1} \in \Out(y_t)\), so condition~(1) holds.

Now assume \(y_t \ne y_{t+1}\). Let \(i=y_t\), and let \(s\) be the first time index of the maximal block of \(i\) ending at time \(t\). Then
\[
y_s = y_{s+1} = \cdots = y_t = i,
\qquad
y_{s-1} \ne i
\quad \text{if } s>0.
\]
When the machine enters visible state \(i\) at time \(s\), its residual phase is reset to \(d_i-1\). Each subsequent time step spent in \(i\) decreases this phase by one. A non-self departure at time \(t\) is possible only if the symbol read at time \(t+1\) is processed from phase \(0\). Thus the phase must have gone from \(d_i-1\) down to \(0\) over the interval from entry at time \(s\) to the departure point at time \(t\), which requires at least \(d_i-1\) steps after entry:
\[
t-s \ge d_i-1.
\]
Equivalently,
\[
t-s+1=\run_{\tau}(t)\ge d_i,
\]
which is condition~(2).

For sufficiency, assume that \(\tau\) satisfies conditions~(1) and~(2). We construct an input word \(w=u_1\cdots u_n\) inductively so that the induced visible trace is exactly \(\tau\).

At time \(t\), suppose the currently realized visible state is \(y_t\). Let
\[
r_t := \max\{d_{y_t}-\run_{\tau}(t),\, 0\}.
\]
This is the residual phase that should be present at time \(t\) if the earlier part of the trace has already been realized. There are three cases.

\smallskip
\noindent
\emph{Case 1: \(y_{t+1} \ne y_t\).}
By condition~(2), \(\run_{\tau}(t) \ge d_{y_t}\), hence \(r_t=0\). Since the graph is destination-complete and \(y_{t+1} \in \Out(y_t)\), there exists an input symbol \(u_{t+1}\) such that
\[
g_{y_t}(u_{t+1}) = \ell(y_{t+1}).
\]
Reading this symbol causes the visible state to move to \(y_{t+1}\), as required.

\smallskip
\noindent
\emph{Case 2: \(y_{t+1}=y_t\) and \(r_t>0\).}
Then the dynamics force the visible state to remain equal to \(y_t\) regardless of the input. Choose any \(u_{t+1} \in \Omega\).

\smallskip
\noindent
\emph{Case 3: \(y_{t+1}=y_t\) and \(r_t=0\).}
Here the machine is free to decide, but we want it to stay at the same visible state. Because the visible graph is self-looping and destination-labeled, the self-label \(\ell(y_t)\) belongs to \(\Sigma_{y_t}\), and by destination-completeness there exists an input symbol \(u_{t+1}\) such that
\[
g_{y_t}(u_{t+1}) = \ell(y_t).
\]
Reading this input triggers the self-loop and again realizes \(y_{t+1}=y_t\).

In every case we can choose \(u_{t+1}\) so that the prefix \(u_1\cdots u_{t+1}\) realizes the prefix \(y_0\cdots y_{t+1}\). Induction on \(t\) completes the construction. \qedhere
\end{proof}

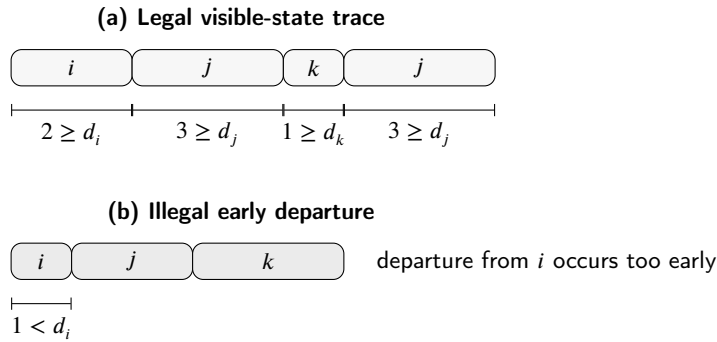
\begin{figure}[pos=H]
\centering
\begin{tikzpicture}[
    x=8mm,y=8mm,
    >=Latex,
    lab/.style={font=\small},
    good/.style={draw,rounded corners,fill=black!3},
    bad/.style={draw,rounded corners,fill=black!8}
]
% Legal trace
\node[font=\small\bfseries] at (3.8,2.2) {(a) Legal visible-state trace};
\draw[good] (0,1.1) rectangle (2,1.7);
\draw[good] (2,1.1) rectangle (4.5,1.7);
\draw[good] (4.5,1.1) rectangle (5.5,1.7);
\draw[good] (5.5,1.1) rectangle (8.0,1.7);
\node at (1,1.4) {$i$};
\node at (3.25,1.4) {$j$};
\node at (5.0,1.4) {$k$};
\node at (6.75,1.4) {$j$};
\draw[|-|] (0,0.7) -- (2,0.7);
\node[lab] at (1,0.3) {$2\ge d_i$};
\draw[|-|] (2,0.7) -- (4.5,0.7);
\node[lab] at (3.25,0.3) {$3\ge d_j$};
\draw[|-|] (4.5,0.7) -- (5.5,0.7);
\node[lab] at (5.0,0.3) {$1\ge d_k$};
\draw[|-|] (5.5,0.7) -- (8.0,0.7);
\node[lab] at (6.75,0.3) {$3\ge d_j$};

% Illegal trace
\node[font=\small\bfseries] at (3.8,-1.0) {(b) Illegal early departure};
\draw[bad] (0,-2.1) rectangle (1.0,-1.5);
\draw[bad] (1.0,-2.1) rectangle (3.0,-1.5);
\draw[bad] (3.0,-2.1) rectangle (5.5,-1.5);
\node at (0.5,-1.8) {$i$};
\node at (2.0,-1.8) {$j$};
\node at (4.25,-1.8) {$k$};
\draw[|-|] (0,-2.5) -- (1.0,-2.5);
\node[lab] at (0.5,-2.9) {$1<d_i$};
\node[lab,anchor=west,align=left] at (5.9,-1.8) {departure from $i$ occurs too early};
\end{tikzpicture}
\caption{Visible-state traces as run-length constraints. \autoref{thm:trace_characterization} can be read directly on the visible trace: each maximal block of a visible state must have length at least the dwell value attached to that state before a departure to a different visible state is allowed.}
\label{fig:visible_trace_runlength}
\end{figure}

\begin{cor}[Regularity of the visible-trace language]
\label{cor:regularity}
Assume destination-completeness. Then
\[
L_{\mathrm{tr}}(\mathcal A)
:=
\{\tr_{\mathcal A}(w) : w \in \Omega^*\}
\subseteq S^*
\]
is a regular language.
\end{cor}

\begin{proof}
By \autoref{thm:compilation_transducer}, the symbolic DLSL system is realized by a finite-state transducer on the phase-expanded state set. By \autoref{thm:trace_characterization}, the visible traces are exactly those visible-state words satisfying a finite-memory condition: graph-feasible adjacency together with a minimum dwell requirement before each departure. Either description yields a finite automaton recognizing the trace language. Hence the language is regular \cite{HopcroftMotwaniUllman2006}. \qedhere
\end{proof}

\begin{cor}[Injective-label version]
If the visible-state labeling \(\ell:S\to\Sigma\) is injective, then minimum dwell is equivalent to a state-dependent minimum run-length condition on the emitted label stream \((\ell(y_t))\).
\end{cor}

\begin{proof}
When \(\ell\) is injective, maximal visible-state blocks and maximal emitted-label blocks coincide. Apply \autoref{thm:trace_characterization}. \qedhere
\end{proof}

\bibliographystyle{unsrtnat}
\bibliography{Bibliographie}

\end{document}